\documentclass[3p]{els2article}

\usepackage[utf8]{inputenc}

\usepackage{graphicx}
\usepackage{graphics}
\usepackage{latexsym}
\usepackage{textcomp}
\usepackage{amsmath}
\usepackage{amsfonts}
\usepackage{amssymb}
\usepackage{amsthm}
\usepackage{bussproofs}
\usepackage{stmaryrd}
\usepackage[all]{xy}
\usepackage{cmll}

\usepackage{boxedminipage}
\usepackage{tabularx}
\usepackage{url}
\usepackage{proof}
\usepackage{bussproofs}
\usepackage{latexsym}
\usepackage{stmaryrd}
\usepackage{enumerate}
\usepackage{pdflscape}

\usepackage{multicol}

\def\bs{\boldsymbol}

\def\ob{\langle}
\def\cb{\rangle}
\def\ra{\rightarrow}

\def\m{\mathcal}

\def\cp{\preccurlyeq}

\DeclareSymbolFont{symbolsC}{U}{txsyc}{m}{n}
\DeclareMathSymbol{\strictif}{\mathrel}{symbolsC}{74}
\DeclareMathSymbol{\boxright}{\mathrel}{symbolsC}{128}
\DeclareMathSymbol{\Diamondright}{\mathrel}{symbolsC}{132}
\DeclareMathSymbol{\boxRight}{\mathrel}{symbolsC}{136}
\DeclareMathSymbol{\DiamondRight}{\mathrel}{symbolsC}{140}
\DeclareMathSymbol{\Diamonddot}{\mathrel}{symbolsC}{144}

\newenvironment{dedsystem}
{\centering
 \newcolumntype{Y}{>{\small\centering\arraybackslash}X}
 \setlength{\extrarowheight}{4ex}
 \tabularx{\textwidth}{|YYY|}}
{\endtabularx}

\newdefinition{definition}{Definition}

\newtheorem{theorem}{Theorem}
\newtheorem{lemma}{Lemma}
\newtheorem{corollary}{Corollary}

\title{Intuitionistic PUC-Logic for Constructive Counterfactuals}

\author[cds]{Ricardo Queiroz de Araujo Fernandes\fnref{fn1}}                                                                            
\ead{ricardo@cds.eb.mil.br}

\author[di]{Edward Hermann Haeusler}                   
\ead{hermann@inf.puc-rio.br}     

\author[df]{Luiz Carlos Pereira}                   
\ead{luiz@inf.puc-rio.br}                                       
                                                                                            
\fntext[fn1]{Thanks to PUC-Rio for the VRac sponsor. Thanks to DAAD (Germany) for the Specialist Literature Programme.}

\address[cds]{Centro de Desenvolvimento de Sistemas\\QG do Exército, Bloco G, 2º Andar\\Setor Militar Urbano, CEP 70630-901\\Brasília, DF, Brasil}

\address[di]{Departamento de Inform\'atica}

\address[df]{Departamento de Filosofia\\Pontif\'icia Universidade Cat\'olica do Rio de Janeiro\\Rua Marquês de São Vicente, 225, Gávea, CEP 22453-900\\Rio de Janeiro, RJ, Brasil}
                                                             
\begin{document}

\begin{abstract} 
We present the intuitionistic version of PUC-Logic. After that, we present a constructive approach to Lewis' counterfactual abstraction to show that it does not require the classical absurd rule.
\end{abstract}


\maketitle

\section*{Introduction}

In human speech, one may change his opinion about the truth of a given counterfactual sentence if his knowledge about the subject grows. For example:

\begin{itemize}
\item[--] (Pedro) I found Ana sad during the party.
\item[--] (Jonas) If her boyfriend had come to party, then Ana would have been happy during the party.
\end{itemize}

\noindent A day after this speech, Jonas discovered that Ana was sad during the party, because his boyfriend betrayed her. So, he may think that his counterfactual sentence is no longer valid.\\

The intuitionistic approach is a traditional way to deal with knowledge growth. We discuss one alternative approach over the counterfactual logic to express this property of the human speech, presenting the intuitionistic reformulation of the PUC-logic, called iPUC-Logic for short.\\

We avoid the repetition in this chapter of definitions and lemmas of PUC-Logic \cite{PUC-Logic-ArXiv} that are the same for iPUC-Logic.

\begin{definition} \label{I_accessibility}
Given a set of worlds $\m{W}$, a nested sets function $\$$ over $\m{W}$ and a truth evaluation function $\m{V}$ for each atomic formula, we define a relation of \textit{accessibility} from a world $u$ to a world $v$, denoted by $u \Yright v$, as a reflexive and transitive relation such that $\$(v) = \$(u)$ and, if $u \in \m{V}(\alpha)$ and $u \Yright v$, then $v \in \m{V}(\alpha)$. Given a world $u$, the set of worlds $v$, such that $u \Yright v$, is denoted by $\m{A}(u)$.
\end{definition}

The restriction $\$(v) = \$(u)$ means that the proximity notions are preserved by the accessible worlds. This restriction was meant to preserve lemma \ref{I_lemmaRa}. But since we are interested in knowledge growth, it is not an artificial restriction. In the example above, the perception of Jonas about how things works did not changed much from the additional knowledge. It means that, for him, the notions of similarity were just the same as in the day before.

\begin{definition} \label{I_modellingDefinition}
Given a variable assignment function $\sigma$, the relation $\models$ of \textit{satisfaction} between wff, labels, models and templates is given by:
\begin{enumerate}
\setcounter{enumi}{1}
\item $\ob \m{W} , \$ , \m{V} , \chi \cb \models \neg \; (\alpha^{\Sigma})$ iff: $\neg \; (\alpha^{\Sigma}) \in \bs{F}_{n}$ and\\\hspace*{4cm}$\forall \lambda \in \m{A}(\chi) : \ob \m{W} , \$ , \m{V} , \lambda \cb \not\models \alpha^{\Sigma}$;

\setcounter{enumi}{4}
\item $\ob \m{W} , \$ , \m{V} , \chi \cb \models \alpha^{\Sigma} \ra \beta^{\Omega}$ iff: $\alpha^{\Sigma} \ra \beta^{\Omega} \in \bs{F}_{n}$ and\\\hspace*{4,4cm}$\forall \lambda \in \m{A}(\chi)$, if $\ob \m{W} , \$ , \m{V} , \lambda \cb \models \alpha^{\Sigma}$, then $\ob \m{W} , \$ , \m{V} , \lambda \cb \models \beta^{\Omega}$;

\setcounter{enumi}{13}
\item $\ob \m{W} , \$ , \m{V} , \chi , N \cb \models \neg \; (\alpha^{\Sigma})$ iff: $\neg \; (\alpha^{\Sigma}) \in \bs{F}_{w}$ and $\forall \lambda \in \m{A}(\chi) : \ob \m{W} , \$ , \m{V} , \lambda , N \cb \not\models \alpha^{\Sigma}$;

\setcounter{enumi}{16}
\item $\ob \m{W} , \$ , \m{V} , \chi , N \cb \models \alpha^{\Sigma} \ra \beta^{\Omega}$ iff: $\alpha^{\Sigma} \ra \beta^{\Omega} \in \bs{F}_{w}$ and\\\hspace*{4,85cm}$\forall \lambda \in \m{A}(\chi)$: if $\ob \m{W} , \$ , \m{V} , \lambda , N \cb \models \alpha^{\Sigma}$, then $\ob \m{W} , \$ , \m{V} , \lambda , N \cb \models \beta^{\Omega}$.
\end{enumerate}
\end{definition}

The iPUC Natural Deduction System is obtained from the PUC-ND by removing the rule 7 (classical absurd rule).

\section{iPUC Soundness and Completeness}

\begin{lemma} \label{I_lemmaRa}
Given $\Delta$ without existential quantifiers, if $(\alpha^{\Sigma} \ra \beta^{\Omega})^{\overline{\Delta}}$ is wff, then it implies $\alpha^{\Sigma,\overline{\Delta}} \ra \beta^{\Omega,\overline{\Delta}}$.
\end{lemma}

\begin{proof} We proceed by induction on the size of $\Delta$:\\
\noindent If $\Delta$ is empty, then the implication is true;\\
\noindent (base) If $\Delta$ contains only one label, it must be a neighbourhood label:
\begin{itemize}
\item[-] $(\alpha^{\Sigma} \ra \beta^{\Omega})^{\circledast}$ means, by definition, that $\forall N \in \$(\chi) : \ob \m{W} , \$ , \m{V} , \chi , N \cb \models \alpha^{\Sigma} \ra \beta^{\Omega}$. Then we know that $\forall N \in \$(\chi) : \forall \lambda \in \m{A}(\chi) :$ if $\ob \m{W} , \$ , \m{V} , \lambda , N \cb \models \alpha^{\Sigma}$, then $\ob \m{W} , \$ , \m{V} , \lambda , N \cb \models \beta^{\Omega}$. Given an arbitrary $\lambda \in \m{A}(\chi)$, if $\ob \m{W} , \$ , \m{V} , \lambda \cb \models \alpha^{\Sigma,\circledast}$, then $\forall M \in \$(\lambda) : \ob \m{W} , \$ , \m{V} , \lambda , M \cb \models \alpha^{\Sigma}$. Since $\$(\lambda) = \$(\chi)$, by the definition of the accessibility relation, we have $\forall M \in \$(\chi) : \ob \m{W} , \$ , \m{V} , \lambda , M \cb \models \alpha^{\Sigma}$. So, by a conclusion above, we know that $\forall M \in \$(\chi) : \ob \m{W} , \$ , \m{V} , \lambda , M \cb \models \beta^{\Omega}$ and $\forall M \in \$(\lambda) : \ob \m{W} , \$ , \m{V} , \lambda , M \cb \models \beta^{\Omega}$, then $\ob \m{W} , \$ , \m{V} , \lambda \cb \models \beta^{\Omega,\circledast}$. In other words, $\alpha^{\Sigma,\circledast} \ra \beta^{\Omega,\circledast}$;
\item[-] $(\alpha^{\Sigma} \ra \beta^{\Omega})^{N}$ means, by definition, that $\sigma(N) \in \$(\chi)$ and $\ob \m{W} , \$ , \m{V} , \chi , \sigma(N) \cb \models \alpha^{\Sigma} \ra \beta^{\Omega}$. Then we know that $\forall \lambda \in \m{A}(\chi) :$ if $\ob \m{W} , \$ , \m{V} , \lambda , \sigma(N) \cb \models \alpha^{\Sigma}$, then $\ob \m{W} , \$ , \m{V} , \lambda , \sigma(N) \cb \models \beta^{\Omega}$. Given an arbitrary $\lambda \in \m{A}(\chi)$, if $\ob \m{W} , \$ , \m{V} , \lambda \cb \models \alpha^{\Sigma,N}$, then $\sigma(N) \in \$(\lambda)$ and $\ob \m{W} , \$ , \m{V} , \lambda , \sigma(N) \cb \models \alpha^{\Sigma}$. Since $\$(\lambda) = \$(\chi)$, we know that $\sigma(N) \in \$(\chi)$ and by a conclusion above, we know that $\ob \m{W} , \$ , \m{V} , \lambda , \sigma(N) \cb \models \beta^{\Omega}$, so $\ob \m{W} , \$ , \m{V} , \lambda \cb \models \beta^{\Omega,N}$. In other words, $\alpha^{\Sigma,N} \ra \beta^{\Omega,N}$;
\end{itemize}

\noindent (base) If $\Delta$ contains two labels, it may be $\{\circledast,\ast\}$, $\{N,\ast\}$, $\{\circledast,u\}$ or $\{N,u\}$. But we just need to look at the distributivity for the $\ast$ label and for world variables, because we have already seen the distributivity of the $\ra$ connective for the label $\circledast$ and for any neighbourhood variable.
\begin{itemize}
\item[-] $(\alpha^{\Sigma} \ra \beta^{\Omega})^{\ast,\circledast}$ means, by definition, that $\forall N \in \$(\chi) : \forall w \in N : \ob \m{W} , \$ , \m{V} , w \cb \models \alpha^{\Sigma} \ra \beta^{\Omega}$. Then we know that $\forall N \in \$(\chi) : \forall w \in N : \forall \lambda \in \m{A}(w) : \mbox{ if } \ob \m{W} , \$ , \m{V} , \lambda \cb \models \alpha^{\Sigma} \mbox{ then } \ob \m{W} , \$ , \m{V} , \lambda \cb \models \beta^{\Omega}$. Given an arbitrary $\lambda \in \m{A}(\chi)$, if we have $\ob \m{W} , \$ , \m{V} , \lambda \cb \models \alpha^{\Sigma,\ast,\circledast}$, then $\forall M \in \$(\lambda) : \forall z \in M : \ob \m{W} , \$ , \m{V} , z \cb \models \alpha^{\Sigma}$. Since $\$(\lambda) = \$(\chi)$, $\forall M \in \$(\chi) : \forall z \in M : \ob \m{W} , \$ , \m{V} , z \cb \models \alpha^{\Sigma}$. From $z \in \m{A}(z)$ and a conclusion above, we know that $\forall M \in \$(\chi) : \forall z \in M : \ob \m{W} , \$ , \m{V} , z \cb \models \beta^{\Omega}$ and $\forall M \in \$(\lambda) : \forall z \in M : \ob \m{W} , \$ , \m{V} , z \cb \models \beta^{\Omega}$, so $\ob \m{W} , \$ , \m{V} , \lambda \cb \models \beta^{\Omega,\ast,\circledast}$. It means that $\alpha^{\Sigma,\ast,\circledast} \ra \beta^{\Omega,\ast,\circledast}$;
\item[-] The proofs of $(\alpha^{\Sigma} \ra \beta^{\Omega})^{\ast,N}$, $(\alpha^{\Sigma} \ra \beta^{\Omega})^{u,\circledast}$ and $(\alpha^{\Sigma} \ra \beta^{\Omega})^{u,N}$ are analogous.
\end{itemize}

\noindent (induction) If $\alpha^{\Sigma} \ra \beta^{\Omega} \in \bs{F}_{w}$ and $s(\Delta) = n + 1$, then the scope must be a neighbourhood label and $\alpha^{\Sigma,\overline{\Delta}} \ra \beta^{\Omega,\overline{\Delta}}$ may be written as $\alpha^{\Sigma,\phi,\overline{\Delta'}} \ra \beta^{\Omega,\phi,\overline{\Delta'}}$, where $s(\Delta') = n$. Then, by the induction hypothesis, $\alpha^{\Sigma,\phi,\overline{\Delta'}} \ra \beta^{\Omega,\phi,\overline{\Delta'}} = (\alpha^{\Sigma,\phi} \ra \beta^{\Omega,\phi})^{\overline{\Delta'}}$. From the base assertions, $(\alpha^{\Sigma,\phi} \ra \beta^{\Omega,\phi})^{\overline{\Delta'}} = ((\alpha^{\Sigma} \ra \beta^{\Omega})^{\phi})^{\overline{\Delta'}} = (\alpha^{\Sigma} \ra \beta^{\Omega})^{\phi,\overline{\Delta'}} = (\alpha^{\Sigma} \ra \beta^{\Omega})^{\overline{\Delta}}$;\\
\noindent (induction) If $\alpha^{\Sigma} \ra \beta^{\Omega} \in \bs{F}_{n}$ and $s(\Delta) = n + 2$, then the scope must be a world label and $\alpha^{\Sigma,\overline{\Delta}} \ra \beta^{\Omega,\overline{\Delta}}$ may be written as $\alpha^{\Sigma,\phi,\Theta,\overline{\Delta'}} \ra \beta^{\Omega,\phi,\Theta,\overline{\Delta'}}$, where $s(\Delta') = n$. Then, by the induction hypothesis, $\alpha^{\Sigma,\phi,\Theta,\overline{\Delta'}} \ra \beta^{\Omega,\phi,\Theta,\overline{\Delta'}} = (\alpha^{\Sigma,\phi,\Theta} \ra \beta^{\Omega,\phi,\Theta})^{\overline{\Delta'}}$. By the base, $(\alpha^{\Sigma,\phi,\Theta} \ra \beta^{\Omega,\phi,\Theta})^{\overline{\Delta'}} = ((\alpha^{\Sigma} \ra \beta^{\Omega})^{\phi,\Theta})^{\overline{\Delta'}} = (\alpha^{\Sigma} \ra \beta^{\Omega})^{\phi,\Theta,\overline{\Delta'}} = (\alpha^{\Sigma} \ra \beta^{\Omega})^{\overline{\Delta}}$.\end{proof}

\begin{lemma} \label{I_resolution}
iPUC-ND without the rules $5, 11, 18, 20, 27, 28$ and $29$ preserves resolution.
\end{lemma}

\begin{proof} Consider $\m{M} = \ob \m{W} , \$ , \m{V} , \chi \cb$. We present the proof for rule 12, because the proof for the other rules are equal as in lemma 8 of \cite{PUC-Logic-ArXiv}.\\
\noindent (12) If $\m{M} \models^{\Delta} \alpha^{\Sigma} \ra \beta^{\Omega}$, then $\m{M} \models (\alpha^{\Sigma} \ra \beta^{\Omega})^{\overline{\Delta}}$, then, by lemma \ref{I_lemmaRa}, $\m{M} \models \alpha^{\Sigma,\overline{\Delta}} \ra \beta^{\Omega,\overline{\Delta}}$. Then, by definition, $\forall \lambda \in \m{A}(\chi) : \mbox{ if } \ob \m{W} , \$ , \m{V} , \lambda \cb \models \alpha^{\Sigma,\overline{\Delta}}, \mbox { then } \ob \m{W} , \$ , \m{V} , \lambda \cb \models \beta^{\Omega,\overline{\Delta}}$. Since $\chi \in \m{A}(\chi)$, then, if $\ob \m{W} , \$ , \m{V} , \chi \cb \models \alpha^{\Sigma,\overline{\Delta}}, \mbox { then } \ob \m{W} , \$ , \m{V} , \chi \cb \models \beta^{\Omega,\overline{\Delta}}$. But we already know that $\ob \m{W} , \$ , \m{V} , \chi \cb \models \alpha^{\Sigma,\overline{\Delta}}$ and we can conclude $\m{M} \models^{\Delta} \beta^{\Omega}$.\end{proof}

\begin{lemma} \label{roadRunner}
If $\ob \m{W} , \$ , \m{V} , \chi \cb \models \alpha^{\Sigma}$, then $\forall \lambda \in \m{A}(\chi) : \ob \m{W} , \$ , \m{V} , \lambda \cb \models \alpha^{\Sigma}$
\end{lemma}

\begin{proof}
If $\ob \m{W} , \$ , \m{V} , \chi \cb \models \alpha$, $\alpha$ atomic, then, by the definition of the accessibility relation, then $\forall \lambda \in \m{A}(\chi) : \ob \m{W} , \$ , \m{V} , \lambda \cb \models \alpha$. If $\ob \m{W} , \$ , \m{V} , \chi \cb \models \alpha^{\Sigma}$, $\Sigma$ non empty, then $\forall \lambda \in \m{A}(\chi) : \ob \m{W} , \$ , \m{V} , \lambda \cb \models \alpha^{\Sigma}$ because $\forall \lambda \in \m{A}(\chi) : \$(\chi) = \$(\lambda)$. The systems of neighbourhoods being the same makes equal the evaluation of the formula. So, if $\ob \m{W} , \$ , \m{V} , \chi \cb \models \alpha^{\Sigma}$, then $\forall \lambda \in \m{A}(\chi) : \ob \m{W} , \$ , \m{V} , \lambda \cb \models \alpha^{\Sigma}$ because the evaluation of all subformulas of $\alpha^{\Sigma}$ are the same at all worlds of $\m{A}(\chi)$, including $\chi$ itself.
\end{proof}

\begin{lemma} \label{I_completeResolutionPUC-ND}
iPUC-ND preserves resolution.
\end{lemma}

\begin{proof} During the proof $\m{M} = \ob \m{W} , \$ , \m{V} , \chi \cb$. We present the proof for each remaining rule of the iPUC-ND inside an induction. Base argument:
\begin{enumerate}
\setcounter{enumi}{4}
\item If $\m{M} \models^{\Delta} \alpha^{\Sigma} \vee \beta^{\Omega}$, then $\m{M} \models (\alpha^{\Sigma} \vee \beta^{\Omega})^{\overline{\Delta}}$, then, by lemma 5 of \cite{PUC-Logic-ArXiv}, $\m{M} \models \alpha^{\Sigma,\overline{\Delta}} \vee \beta^{\Omega,\overline{\Delta}}$, then, by definition, $\m{M} \models \alpha^{\Sigma,\overline{\Delta}}$ or $\m{M} \models \beta^{\Omega,\overline{\Delta}}$. This means, by definition, that $\m{M} \models^{\Delta} \alpha^{\Sigma}$ or $\m{M} \models^{\Delta} \beta^{\Omega}$. So, if $\Pi_{1}$ and $\Pi_{2}$ only contains the rules from lemma \ref{I_resolution}, $\m{M} \models^{\Theta} \gamma^{\Lambda}$ in both cases, because of the preservation of the resolution relation. And, for that conclusion, the hypothesis are no longer necessary and may be discharged;
\setcounter{enumi}{10}
\item If $\Pi$ only contains the rules of lemma \ref{I_resolution}, then, from the hypothesis that $\m{M} \models^{\Delta} \alpha^{\Sigma}$, the derivation gives us $\m{M} \models^{\Delta} \beta^{\Omega}$. If $\beta^{\Omega} \in \bs{F}_{n}$, then, by the fitting relation and lemma 2 of \cite{PUC-Logic-ArXiv}, we know that $s(\Delta)$ is even. If we take some model $\m{H} = \ob \m{W} , \$ , \m{V} , z \cb$, such that $\m{M} \multimap_{s(\Delta)} \m{H}$ and $\m{H} \models \beta^{\Omega}$, then, by lemma \ref{roadRunner}, $\forall w \in \m{A}(z) : \ob \m{W} , \$ , \m{V} , w \cb \models \beta^{\Omega}$ and, by definition, $\m{H} \models \alpha^{\Sigma} \ra \beta^{\Omega}$. So, by definition, $\beta^{\Omega} \models_{\m{M}:s(\Delta)} \alpha^{\Sigma} \ra \beta^{\Omega}$, which means, by lemma 21 of \cite{PUC-Logic-ArXiv}, that $\m{M} \models^{\Delta} \alpha^{\Sigma} \ra \beta^{\Omega}$. If $\beta^{\Omega} \in \bs{F}_{w}$, then, by the fitting relation and lemma 2 of \cite{PUC-Logic-ArXiv}, we know that $s(\Delta)$ is odd. If we take some template $\m{T} = \ob \m{W} , \$ , \m{V} , z , L\cb$, such that $\m{M} \multimap_{s(\Delta)} \m{T}$ and $\m{T} \models \beta^{\Omega}$, then, by lemma \ref{roadRunner}, $\forall w \in \m{A}(z) : \ob \m{W} , \$ , \m{V} , w , L\cb \models \beta^{\Omega}$ and, by definition, $\m{T} \models \alpha^{\Sigma} \ra \beta^{\Omega}$. So, by definition, $\beta^{\Omega} \models_{\m{M}:s(\Delta)} \alpha^{\Sigma} \ra \beta^{\Omega}$, which means, by lemma 21 of \cite{PUC-Logic-ArXiv}, that $\m{M} \models^{\Delta} \alpha^{\Sigma} \ra \beta^{\Omega}$. So the hypothesis is unnecessary and may be discharged;
\setcounter{enumi}{17}
\item If $\m{M} \models^{\Delta,\bullet} \alpha^{\Sigma}$, then, by the rule 14, $\m{M} \models^{\Delta} \alpha^{\Sigma,\bullet}$. From $\alpha^{\Sigma,\bullet} \in \bs{F}_{w}$, the fitting relation and lemma 2 of \cite{PUC-Logic-ArXiv}, we know that $s(\Delta)$ is odd. If we take some template $\m{T} = \ob \m{W} , \$ , \m{V} , z , N \cb$, such that $\m{M} \multimap_{s(\Delta)} \m{T}$ and $\m{T} \models \alpha^{\Sigma,\bullet}$, then, $N \in \$(z)$ and $\exists w \in N : \ob \m{W} , \$ , \m{V} , w \cb \models \alpha^{\Sigma}$. Since the variable $u$ occurs nowhere else in the derivation, $u$ can be taken as a denotation of the given existential and we conclude that $\ob \m{W} , \$ , \m{V} , u \cb \models \alpha^{\Sigma}$, what means that $\m{T} \models \alpha^{\Sigma,u}$. So, by definition, $\alpha^{\Sigma,\bullet} \models_{\m{M}:s(\Delta)} \alpha^{\Sigma,u}$, which means, by lemma 21 of \cite{PUC-Logic-ArXiv}, that $\m{M} \models^{\Delta} \alpha^{\Sigma,u}$. We conclude, using the rule 13, that $\m{M} \models^{\Delta,u} \alpha^{\Sigma}$. If $\Pi$ only contains rules of the lemma \ref{I_resolution}, then we can conclude $\m{M} \models^{\Theta} \beta^{\Omega}$. Then we can discharge the hypothesis because we know that any denotation of the existential may provide the same conclusion;
\setcounter{enumi}{19}
\item If $\m{M} \models^{\Delta,\circledcirc} \alpha^{\Sigma}$, then, by the rule 14, $\m{M} \models^{\Delta} \alpha^{\Sigma,\circledcirc}$. From $\alpha^{\Sigma,\circledcirc} \in \bs{F}_{n}$, the fitting relation and lemma 2 of \cite{PUC-Logic-ArXiv}, we know that $s(\Delta)$ is even. If we take some model $\m{H} = \ob \m{W} , \$ , \m{V} , z \cb$, such that $\m{M} \multimap_{s(\Delta)} \m{H}$ and $\m{H} \models \alpha^{\Sigma,\circledcirc}$, then $\exists M \in \$(z) : \ob \m{W} , \$ , \m{V} , z , M \cb \models \alpha^{\Sigma}$. Since the variable $N$ occurs nowhere else in the derivation, $N$ can be taken as a denotation of the given existential and we conclude that $\ob \m{W} , \$ , \m{V} , z , N \cb \models \alpha^{\Sigma}$, what means that $\m{H} \models \alpha^{\Sigma,N}$. So, by definition, $\alpha^{\Sigma,\circledcirc} \models_{\m{M}:s(\Delta)} \alpha^{\Sigma,N}$, which means, by lemma 21 of \cite{PUC-Logic-ArXiv}, that $\m{M} \models^{\Delta} \alpha^{\Sigma,N}$. We conclude, using the rule 13, that $\m{M} \models^{\Delta,N} \alpha^{\Sigma}$. If $\Pi$ only contains rules of the lemma \ref{I_resolution}, then we can conclude $\m{M} \models^{\Theta} \beta^{\Omega}$. Then we can discharge the hypothesis because we know that any denotation of the existential may provide the same conclusion.
\end{enumerate}
\noindent Inductive case: the same argument as in lemma 28 of \cite{PUC-Logic-ArXiv}.\end{proof}

\begin{theorem} \label{I_soundness}
$\Gamma \vdash \alpha^{\Sigma}$ implies $\Gamma \models \alpha^{\Sigma}$ (Soundness).
\end{theorem}

\begin{proof}
The fitting restriction of the rules of iPUC-ND ensures that $\alpha^{\Sigma}$ has neighbourhood characteristic because it appears in the empty context. The same conclusion follows for every formula of $\Gamma$. The derivability assures that there is a derivation that concludes $\alpha^{\Sigma}$ and takes as open hypothesis a subset of $\Gamma$, which we call $\Gamma'$. If we take a model $\m{M}$ that satisfies every formula of $\Gamma$, then it also satisfies every formula of $\Gamma'$. So, $\m{M} \models \gamma^{\Theta}$, for every $\gamma^{\Theta} \in \Gamma'$. But this means, by definition, that for every wff of $\Gamma'$ the resolution relation holds with the empty context. Then, from lemma \ref{I_completeResolutionPUC-ND}, we know that $\m{M} \models \alpha^{\Sigma}$. So, every model that satisfies every formula of $\Gamma$ also satisfies $\alpha^{\Sigma}$ and, by definition, $\Gamma \models \alpha^{\Sigma}$.\end{proof}

We use prime theories to prove completeness. The reader can see the intuitionistic logic case of this way of proving completeness in \cite{vanDalen}. We recall from \cite{PUC-Logic-ArXiv} the following definitions:

\begin{definition}
Given $\alpha^{\Sigma} \in \bs{F}_{n}$, if $\alpha^{\Sigma}$ has no variables in the attributes of its subformulas nor any subformula of the shape $\shneg N$ or $\shpos N$, then $\alpha^{\Sigma} \in \bs{S}_{n}$. By analogy, we can construct $\bs{S}_{w}$ from $\bs{F}_{w}$.
\end{definition} 

\begin{definition}
Given $\Gamma \subset \bs{S}_{n}$ ($\Gamma \subset \bs{S}_{w}$), we say that $\Gamma$ is \textit{n-inconsistent} (\textit{w-inconsistent}) if $\Gamma \vdash \bot_{n}$ ($\Gamma \vdash^{N}_{N} \bot_{w}$, where $N$ is a neighbourhood variable that does not occur in $\Gamma$) and \textit{n-consistent} (\textit{w-consistent}) if $\Gamma \not\vdash \bot_{n}$ ($\Gamma \not\vdash^{N}_{N} \bot_{w}$).
\end{definition}

\begin{lemma} \label{I_modelForSet}
Given $\Gamma \subset \bs{S}_{n}$ ($\Gamma \subset \bs{S}_{w}$), if there is a model (template) that satisfies every formula of $\Gamma$, then $\Gamma$ is n-consistent.
\end{lemma}

\begin{proof}
This proof is analogous to the classical case, with the addition that here we use theorem \ref{I_soundness} instead of theorem 31 of \cite{PUC-Logic-ArXiv}.\end{proof}

\begin{lemma} \label{I_consistency}
Given $\Gamma \subset \bs{S}_{n}$ ($\Gamma \subset \bs{S}_{w}$), if $\Gamma \cup \{\phi^{\Theta}\} \vdash \bot_{n}$, then $\Gamma \vdash \neg \phi^{\Theta}$.
\end{lemma}

\begin{proof}
The assumption implies that there is a derivation $\m{D}$ with with hypothesis in $\Gamma \cup \{\phi^{\Theta}\}$ and conclusion $\bot_{n}$. If we apply the rule $\ra$-introduction and eliminate all occurrences of $\phi^{\Theta}$ as hypothesis, then we obtain a derivation with hypothesis in $\Gamma$ and conclusion $\neg \phi^{\Theta}$. The same argument holds for $\Gamma \subset \bs{S}_{w}$.\end{proof}

\begin{definition} \label{I_primeTheory}
$\Gamma \subset \bs{S}_{n}$ ($\Gamma \subset \bs{S}_{w}$) is a \textit{prime n-theory} (prime w-theory) iff (i) $\Gamma \vdash \alpha^{\Sigma}$ ($\Gamma \vdash^{N}_{N} \alpha^{\Sigma}$) implies $\alpha^{\Sigma} \in \Gamma$ and (ii) if $\alpha^{\Sigma} \vee \beta^{\Omega} \in \Gamma$ implies $\alpha^{\Sigma} \in \Gamma$ or $\beta^{\Omega} \in \Gamma$.
\end{definition}

\newpage

\begin{lemma} \label{I_subMax}
Given $\Gamma \cup \{\alpha^{\Sigma}\} \subset \bs{S}_{n}$ ($\Gamma \cup \{\alpha^{\Sigma}\} \subset \bs{S}_{w}$), if $\Gamma \not\vdash \alpha^{\Sigma}$ ($\Gamma \not\vdash^{N}_{N} \alpha^{\Sigma}$), then there is a prime n-theory (w-theory) $\Gamma'$, such that $\Gamma \subset \Gamma'$ and $\Gamma' \not\vdash \alpha^{\Sigma}$.
\end{lemma}
\begin{proof}
According to lemma 37 of \cite{PUC-Logic-ArXiv}, we may have a list $\varphi_{0}, \varphi_{1}, \ldots$ of all wff in $\bs{S}_{n}$. We build a non-decreasing sequence of sets $\Gamma_{i}$ such that the union is a prime theory. We put $\Gamma_{0} = \Gamma$. Then we take the first disjunctive sentence that has not been treated and $\Gamma_{n} \vdash \varphi_{0} \vee \varphi_{1}$: $\Gamma_{n+1} = \Gamma_{n} \cup \{\varphi_{0}\}$ if $\Gamma_{n} \cup \{\varphi_{0}\} \not\vdash \alpha^{\Sigma}$, $\Gamma_{n+1} = \Gamma_{n} \cup \{\varphi_{1}\}$ otherwise. It cannot be the case that $\Gamma_{n} \cup \{\varphi_{0}\} \vdash \alpha^{\Sigma}$ and $\Gamma_{n} \cup \{\varphi_{1}\} \vdash \alpha^{\Sigma}$, because, in this case, $\Gamma_{n} \vdash \alpha^{\Sigma}$ by $\vee$-elimination. So, $\Gamma' = \bigcup\{\Gamma_{n} \; | \; n \geq 0 \}$.\\[5pt]
(a) $\Gamma' \not\vdash \alpha^{\Sigma}$: by induction, since $\Gamma_{0} \not\vdash \alpha^{\Sigma}$ and $\Gamma_{n+1} \not\vdash \alpha^{\Sigma}$ by the definition of the induction step. (b) $\Gamma'$ is a prime theory: (i) if $\psi_{1} \vee \psi_{2} \in \Gamma'$, then take the least number $k$ such that $\Gamma_{k} \vdash \psi_{1} \vee \psi_{2}$. So, $\psi_{1} \vee \psi_{2}$ cannot have been treated at a stage before $k$ and $\Gamma_{h} \vdash \psi_{1} \vee \psi_{2}$, for $h \geq k$. At some point $\psi_{1} \vee \psi_{2}$ must be treated at a stage $h \geq k$. Then, $\psi_{1} \in \Gamma_{h+1}$ or $\psi_{2} \in \Gamma_{h+1}$ and, by definition, $\psi_{1} \in \Gamma'$ or $\psi_{2} \in \Gamma'$. (ii) if $\Gamma' \vdash \psi$, then $\Gamma' \vdash \psi \vee \psi$, then by (i) $\psi \in \Gamma'$.\\[5pt] The same argument holds for sentences in $\bs{S}_{w}$.\end{proof}

\begin{definition}
Given a prime n-theory $\Gamma$ and a prime w-theory $\Lambda$, we say that $\Gamma$ \textit{accepts} $\Lambda$ ($\Gamma \propto \Lambda$) if $\alpha^{\Sigma} \in \Lambda$ implies $\alpha^{\Sigma,\circledcirc} \in \Gamma$. If $\alpha^{\Sigma} \in \Gamma$ implies $\alpha^{\Sigma,\bullet} \in \Lambda$, then $\Lambda \propto \Gamma$.
\end{definition}

\begin{definition}
Given prime w-theories $\Gamma$ and $\Lambda$, we say that $\Gamma$ \textit{subordinates} $\Lambda$ ($\Lambda \sqsubset \Gamma$) iff $\alpha^{\Sigma,\bullet} \in \Lambda$ implies $\alpha^{\Sigma,\bullet} \in \Gamma$ and $\alpha^{\Sigma,\ast} \in \Gamma$ implies $\alpha^{\Sigma,\ast} \in \Lambda$.
\end{definition}

\begin{lemma} \label{I_consistentModel}
If $\Gamma \subset \bs{S}_{n}$ is n-consistent, then there is a model $\m{M}$, such that $\m{M} \models \alpha^{\Sigma}$, for every $\alpha^{\Sigma} \in \Gamma$.
\end{lemma}

\begin{proof}
By definition, $\Gamma \not\vdash \bot_{n}$ and, by lemma \ref{I_subMax}, $\Gamma$ is contained in a prime n-theory $\Gamma'$, such that $\Gamma' \not\vdash \bot_{n}$. We take every prime n-theory $\Psi$ as a representation of one world of $\m{W}$, denoted by $\chi_{\Psi}$. Every prime w-theory will be seen as a set of worlds that may be a neighbourhood. We take $\propto$ as the nested neighbourhood function $\$$ and $\sqsubset$ as the total order among neighbourhoods. We take the subset relation among prime n-theories as the accessibility relation among worlds. To build the truth evaluation function $\m{V}$, we require, for every prime n-theory $\Psi$ and for every $\alpha$ atomic: (a) $\chi_{\Psi} \in \m{V}(\alpha)$ if $\alpha \in \Psi$; (b) $\chi_{\Psi} \not\in \m{V}(\alpha)$ if $\alpha \not\in \Psi$. If we take $\m{M} = \ob \m{W} , \$ , \m{V} , \chi_{\Gamma'} \cb$, then, for every wff $\alpha^{\Sigma} \in \Gamma'$, $\m{M} \models \alpha^{\Sigma}$. We proceed by induction on the structure of $\alpha^{\Sigma}$:\\[5pt]
\noindent (Base) If $\alpha^{\Sigma}$ is atomic, $\m{M} \models \alpha^{\Sigma}$ iff $\alpha^{\Sigma} \in \Gamma'$, by the definition of $\m{V}$;\\[5pt]
\noindent $\alpha^{\Sigma} = \beta^{\Omega} \wedge \gamma^{\Theta}$. $\m{M} \models \alpha^{\Sigma}$ iff $\m{M} \models \beta^{\Omega}$ and $\m{M} \models \gamma^{\Theta}$ iff (induction hypothesis) $\beta^{\Omega} \in \Gamma'$ and $\gamma^{\Theta} \in \Gamma'$. We conclude that $\alpha^{\Sigma} \in \hat{\Gamma}$ by a single application of the $\wedge$-introduction rule and the fact that a prime theory is closed by derivability. Conversely $\alpha^{\Sigma} \in \Gamma'$ iff $\beta^{\Omega} \in \Gamma'$ and $\gamma^{\Theta} \in \Gamma'$ by $\wedge$-elimination and the fact that a prime theory is closed by derivability. The rest follows by the induction hypothesis;\\[5pt]
\noindent $\alpha^{\Sigma} = \beta^{\Omega} \vee \gamma^{\Theta}$. By the definition of prime theory, $\beta^{\Omega} \in \Gamma'$ or $\gamma^{\Theta} \in \Gamma'$ and we proceed by induction.\\[5pt]
\noindent $\alpha^{\Sigma} = \beta^{\Omega} \ra \gamma^{\Theta}$. $\m{M} \not\models \alpha^{\Sigma}$ iff $\exists \chi_{\Psi} \in \m{A}(\chi_{\Gamma'})$ such that $\ob \m{W} , \$ , \m{V} , \chi_{\Psi} \cb \models \beta^{\Omega}$ and $\ob \m{W} , \$ , \m{V} , \chi_{\Psi} \cb \not\models \gamma^{\Theta}$ iff (induction hypothesis) $\beta^{\Omega} \in \Psi$ and $\gamma^{\Theta} \not\in \Psi$. So, by the fact that a prime theory is closed by derivability and by $\ra$-elimination, we conclude that $\beta^{\Omega} \ra \gamma^{\Theta} \not\in \Psi$. By the definition of accessibility among worlds, $\Gamma' \subset \Psi$ and $\beta^{\Omega} \ra \gamma^{\Theta} \not\in \Gamma'$. Conversely, $\beta^{\Omega} \ra \gamma^{\Theta} \not\in \Gamma'$ implies that, if $\beta^{\Omega} \in \Gamma'$, then $\gamma^{\Theta} \not\in \Gamma'$, because, on the contrary, for every prime n-theory $\Psi$, such that $\Gamma' \subset \Psi$, $\gamma^{\Theta} \in \Psi$ and by i.h. $\ob \m{W} , \$ , \m{V} , \chi_{\Psi} \cb \models \gamma^{\Theta}$. So, by the definition of accessibility of worlds, for every $\chi_{\Psi} \in \m{A}(\chi_{\Gamma'}) : \ob \m{W} , \$ , \m{V} , \chi_{\Psi} \cb \models \gamma^{\Theta}$ and, by definition, $\m{M} \models \beta^{\Omega} \ra \gamma^{\Theta}$. From $\beta^{\Omega} \in \Gamma'$ we know that by i.h. $\m{M} \models \beta^{\Omega}$. So, from $\chi_{\Gamma'} \in \m{A}(\chi_{\Gamma'})$ and $\m{M} \models \beta^{\Omega} \ra \gamma^{\Theta}$, we get $\m{M} \models \gamma^{\Theta}$ and by i.h. $\gamma^{\Theta} \in \Gamma'$. But, in this case, $\beta^{\Omega} \ra \gamma^{\Theta} \in \Gamma'$ by a single $\ra$introduction and the fact that a prime theory is closed by derivability;\\[5pt]
\noindent $\alpha^{\Sigma} = \neg \beta^{\Omega}$. This case is treated by the previous case, because $\neg \beta^{\Omega} \equiv \beta^{\Omega} \ra \bot_{n}$;\\[5pt]
\noindent $\alpha^{\Sigma} = \beta^{\Omega,\circledast}$. This wff do not require the existence of a prime w-theory to represent a neighbourhood in which $\beta^{\Omega}$ holds, on the contrary, it only requires that there is no prime w-theory, accepted by $\Upsilon$, in which $\gamma^{\Theta}$ does not hold.\\[5pt]
\noindent $\alpha^{\Sigma} = \beta^{\Omega,\circledcirc}$. We take any enumeration $\rho_{1},\rho_{2},\ldots$ in $\bs{F}_{w}$, such that $\rho_{i}^{\circledcirc} \in \Gamma'$ and $\rho_{0} = \beta^{\Omega}$. We construct a w-theory by following:
\begin{itemize}
\item $\Upsilon_{0} = \{\beta^{\Omega}\}$;
\item $\Upsilon_{i} = \Upsilon_{i-1} \cup \{\rho_{i}\}$, if $(\varphi_{0}\wedge\ldots\wedge\varphi_{m}\wedge\rho{i})^{\circledcirc}\in \Gamma'$, where $\Upsilon_{i-1} = \{\varphi_{0},\ldots,\varphi_{m}\}$. $\Upsilon_{i} = \Upsilon_{i-1}$ otherwise;
\item $\Upsilon = \bigcup_{n \in \mathbb{N}}\Upsilon_{n}$.
\end{itemize}
Now we must prove the following:
\begin{enumerate}
\item $\Upsilon$ is a w-theory:
\begin{itemize}
\item $\Upsilon \not\vdash \bot_{w}$. A derivation in iPUC build only with wff from $\bs{F}_{w}$ must have all rules with odd sized context. To do so, we choose a neighbourhood variable, that do not occur during the derivation, to be the leftmost label at every context that appears during the derivation. This new variable may be understood as the representation of the neighbourhood in which the inference is made. So, if $\Upsilon$ is not w-consistent, then, given some derivation $\m{D}$ that concludes $\bot_{w}$ from wff of $\Upsilon$, there is an index $n \in \mathbb{N}$, such that $\Upsilon_{n}$ contains all wff that appears in $\m{D}$. So, by definition, that there is a wff $(\varphi_{0}\wedge\ldots\wedge\varphi_{m})^{\circledcirc} \in \Gamma'$, such that $\varphi_{0},\ldots,\varphi_{m}$ represent all wff of $\Upsilon_{n}$ and $\bot_{n} \in \Gamma'$ by the derivation below, which is a contradiction by the definition of $\Gamma'$.\begin{prooftree}
\AxiomC{$(\varphi_{0}\wedge\ldots\wedge\varphi_{m})^{\circledcirc}$}
\UnaryInfC{$(\varphi_{0}\wedge\ldots\wedge\varphi_{m})^{\circledcirc}$} \RightLabel{$\circledcirc$}
\UnaryInfC{$\varphi_{0}\wedge\ldots\wedge\varphi_{m}$}
\AxiomC{[$\varphi_{0}\wedge\ldots\wedge\varphi_{m}$]} \RightLabel{$N$}
\UnaryInfC{$\m{D}'$} \RightLabel{$N$}
\UnaryInfC{$\bot_{w}$}
\UnaryInfC{$\bot_{n}$}
\BinaryInfC{$\bot_{n}$}\end{prooftree}
The derivation $\m{D}'$ is obtained from $\m{D}$ by 1) adding $\varphi_{0}\wedge\ldots\wedge\varphi_{m}$ in the place of the hypothesis $\varphi_{i}$ and the following $\wedge$-eliminations to recover $\varphi_{i}$; 2) binding the variable $N$ added to produce $\m{D}$.
\item If $\varpi \vee \vartheta \in \Upsilon$, then $\exists i \in \mathbb{N} : \rho_{i} = \varpi \vee \vartheta \in \Upsilon_{i}$ and $(\varphi_{0}\wedge\ldots\wedge\varphi_{m}\wedge(\varpi \vee \vartheta))^{\circledcirc}\in \Gamma'$. By the distribution of $\vee$ over $\wedge$, we know that $((\varphi_{0}\wedge\ldots\wedge\varphi_{m}\wedge\varpi) \vee (\varphi_{0}\wedge\ldots\wedge\varphi_{m}\wedge\vartheta))^{\circledcirc} \in \Gamma'$. By the following derivation ($\Pi_{1}$), we know that $(\varphi_{0}\wedge\ldots\wedge\varphi_{m}\wedge\varpi)^{\circledcirc} \vee (\varphi_{0}\wedge\ldots\wedge\varphi_{m}\wedge\vartheta)^{\circledcirc}\in \Gamma'$ and, by the definition of prime theory, $(\varphi_{0}\wedge\ldots\wedge\varphi_{m}\wedge\varpi)^{\circledcirc} \in \Gamma'$ or $(\varphi_{0}\wedge\ldots\wedge\varphi_{m}\wedge\vartheta)^{\circledcirc}\in \Gamma'$. We proceed now considering the possibilities of the order provided by the enumeration $\rho$. Without loss of generality, we suppose that $\rho_{k} = \varpi$, $\rho_{l} = \vartheta$ and $k < l$. In the case $l < i$, we are done, because, from $(\varphi_{0}\wedge\ldots\wedge\varphi_{m}\wedge\varpi)^{\circledcirc} \in \Gamma'$ or $(\varphi_{0}\wedge\ldots\wedge\varphi_{m}\wedge\vartheta)^{\circledcirc}\in \Gamma'$, we can eliminate the extra formulas to recover $\Upsilon_{k}$ or $\Upsilon_{l}$ by the rule of construction of the set $\Upsilon$. In the case $i < m$, at the moment of verification of $\rho_{m}$, we know that for all formulas of $\Upsilon_{m-1} = \{\varphi_{0},\ldots,(\varpi \vee \vartheta),\ldots,\varphi_{t}\}$ we have a conjunction $(\varphi_{0}\wedge\ldots\wedge(\varpi \vee \vartheta)\wedge\ldots\wedge\varphi_{t})^{\circledcirc} \in \Gamma'$. By a derivation similar to $\Pi_{1}$, we know that $(\varphi_{0}\wedge\ldots\wedge\varpi\wedge\ldots\wedge\varphi_{t})^{\circledcirc} \in \Gamma'$ or $(\varphi_{0}\wedge\ldots\wedge\vartheta\wedge\ldots\wedge\varphi_{t})^{\circledcirc} \in \Gamma'$ and by means of a $\vee$-introduction in another derivation we know that $(\varphi_{0}\wedge\ldots\wedge(\varpi \vee \vartheta)\wedge\ldots\wedge\varphi_{t}\wedge\varpi)^{\circledcirc} \in \Gamma'$ or $(\varphi_{0}\wedge\ldots\wedge(\varpi \vee \vartheta)\wedge\ldots\wedge\varphi_{t}\wedge\vartheta)^{\circledcirc} \in \Gamma'$. If we have $(\varphi_{0}\wedge\ldots\wedge(\varpi \vee \vartheta)\wedge\ldots\wedge\varphi_{t}\wedge\varpi)^{\circledcirc} \in \Gamma'$, then $\Upsilon_{m}$ and $\Upsilon$ contain $\varpi$. If we have $(\varphi_{0}\wedge\ldots\wedge(\varpi \vee \vartheta)\wedge\ldots\wedge\varphi_{t}\wedge\vartheta)^{\circledcirc} \in \Gamma'$, at the moment of verification of $\rho_{l}$, we make the distribution of the disjunction $(\varpi \vee \vartheta)$ again. Given $\Upsilon_{l-1} = \{\varphi_{0},\ldots,(\varpi \vee \vartheta),\ldots,\varphi_{s}\}$ we have a conjunction $(\varphi_{0}\wedge\ldots\wedge(\varpi \vee \vartheta)\wedge\ldots\wedge\varphi_{s})^{\circledcirc} \in \Gamma'$. By a derivation similar to $\Pi_{1}$, we know that $(\varphi_{0}\wedge\ldots\wedge\varpi\wedge\ldots\wedge\varphi_{s})^{\circledcirc} \in \Gamma'$ or $(\varphi_{0}\wedge\ldots\wedge\vartheta\wedge\ldots\wedge\varphi_{s})^{\circledcirc} \in \Gamma'$ and by means of a $\vee$-introduction in another derivation we know that $(\varphi_{0}\wedge\ldots\wedge(\varpi \vee \vartheta)\wedge\ldots\wedge\varphi_{s}\wedge\varpi)^{\circledcirc} \in \Gamma'$ or $(\varphi_{0}\wedge\ldots\wedge(\varpi \vee \vartheta)\wedge\ldots\wedge\varphi_{s}\wedge\vartheta)^{\circledcirc} \in \Gamma'$. If $(\varphi_{0}\wedge\ldots\wedge(\varpi \vee \vartheta)\wedge\ldots\wedge\varphi_{s}\wedge\vartheta)^{\circledcirc} \in \Gamma'$, then we are done. If $(\varphi_{0}\wedge\ldots\wedge(\varpi \vee \vartheta)\wedge\ldots\wedge\varphi_{s}\wedge\varpi)^{\circledcirc} \in \Gamma'$, then we have a contradiction, because we supposed that $(\varphi_{0}\wedge\ldots\wedge(\varpi \vee \vartheta)\wedge\ldots\wedge\varphi_{t}\wedge\varpi)^{\circledcirc} \in \Gamma'$ were not the case and we can recover it by eliminating the extra formulas. The the case $m < i < l$ is analogous.
\end{itemize}
\end{enumerate}

\begin{landscape}
\begin{prooftree}
\AxiomC{$((\varphi_{0}\wedge\ldots\wedge\varphi_{m}\wedge\varpi) \vee (\varphi_{0}\wedge\ldots\wedge\varphi_{m}\wedge\vartheta))^{\circledcirc}$}
\UnaryInfC{$((\varphi_{0}\wedge\ldots\wedge\varphi_{m}\wedge\varpi) \vee (\varphi_{0}\wedge\ldots\wedge\varphi_{m}\wedge\vartheta))^{\circledcirc}$} \RightLabel{$\circledcirc$}
\UnaryInfC{$(\varphi_{0}\wedge\ldots\wedge\varphi_{m}\wedge\varpi) \vee (\varphi_{0}\wedge\ldots\wedge\varphi_{m}\wedge\vartheta)$}
\AxiomC{$(\varphi_{0}\wedge\ldots\wedge\varphi_{m}\wedge\varpi) \vee (\varphi_{0}\wedge\ldots\wedge\varphi_{m}\wedge\vartheta)$} \RightLabel{$N$}
\UnaryInfC{$\Pi_{2}$}
\UnaryInfC{$(\varphi_{0}\wedge\ldots\wedge\varphi_{m}\wedge\varpi)^{\circledcirc} \vee (\varphi_{0}\wedge\ldots\wedge\varphi_{m}\wedge\vartheta)^{\circledcirc}$} \LeftLabel{$\bs{\Pi_{1}}$}
\BinaryInfC{$(\varphi_{0}\wedge\ldots\wedge\varphi_{m}\wedge\varpi)^{\circledcirc} \vee (\varphi_{0}\wedge\ldots\wedge\varphi_{m}\wedge\vartheta)^{\circledcirc}$}
\end{prooftree}

\hfill\\

\begin{prooftree}
\AxiomC{$(\varphi_{0}\wedge\ldots\wedge\varphi_{m}\wedge\varpi) \vee (\varphi_{0}\wedge\ldots\wedge\varphi_{m}\wedge\vartheta)$} \RightLabel{$N$}
\UnaryInfC{$(\varphi_{0}\wedge\ldots\wedge\varphi_{m}\wedge\varpi) \vee (\varphi_{0}\wedge\ldots\wedge\varphi_{m}\wedge\vartheta)$}
\AxiomC{[$\varphi_{0}\wedge\ldots\wedge\varphi_{m}\wedge\varpi$]} \RightLabel{$N$}
\UnaryInfC{$\varphi_{0}\wedge\ldots\wedge\varphi_{m}\wedge\varpi$}
\UnaryInfC{$\Pi_{3}$}
\UnaryInfC{$(\varphi_{0}\wedge\ldots\wedge\varphi_{m}\wedge\varpi)^{\circledcirc} \vee (\varphi_{0}\wedge\ldots\wedge\varphi_{m}\wedge\vartheta)^{\circledcirc}$}
\AxiomC{[$\varphi_{0}\wedge\ldots\wedge\varphi_{m}\wedge\vartheta$]} \RightLabel{$N$}
\UnaryInfC{$\varphi_{0}\wedge\ldots\wedge\varphi_{m}\wedge\vartheta$}
\UnaryInfC{$\Pi_{4}$}
\UnaryInfC{$(\varphi_{0}\wedge\ldots\wedge\varphi_{m}\wedge\varpi)^{\circledcirc} \vee (\varphi_{0}\wedge\ldots\wedge\varphi_{m}\wedge\vartheta)^{\circledcirc}$}
\BinaryInfC{$(\varphi_{0}\wedge\ldots\wedge\varphi_{m}\wedge\varpi)^{\circledcirc} \vee (\varphi_{0}\wedge\ldots\wedge\varphi_{m}\wedge\vartheta)^{\circledcirc}$}\LeftLabel{$\bs{\Pi_{2}}$}
\BinaryInfC{$(\varphi_{0}\wedge\ldots\wedge\varphi_{m}\wedge\varpi)^{\circledcirc} \vee (\varphi_{0}\wedge\ldots\wedge\varphi_{m}\wedge\vartheta)^{\circledcirc}$}
\end{prooftree}

\hfill\\

\begin{prooftree}
\AxiomC{$((\varphi_{0}\wedge\ldots\wedge\varphi_{m}\wedge\varpi) \vee (\varphi_{0}\wedge\ldots\wedge\varphi_{m}\wedge\vartheta))^{\circledcirc}$}
\UnaryInfC{$((\varphi_{0}\wedge\ldots\wedge\varphi_{m}\wedge\varpi) \vee (\varphi_{0}\wedge\ldots\wedge\varphi_{m}\wedge\vartheta))^{\circledcirc}$} \RightLabel{$\circledcirc$}
\UnaryInfC{$((\varphi_{0}\wedge\ldots\wedge\varphi_{m}\wedge\varpi) \vee (\varphi_{0}\wedge\ldots\wedge\varphi_{m}\wedge\vartheta))$}
\AxiomC{$\varphi_{0}\wedge\ldots\wedge\varphi_{m}\wedge\varpi$} \RightLabel{$N$}
\UnaryInfC{$\varphi_{0}\wedge\ldots\wedge\varphi_{m}\wedge\varpi$} \RightLabel{$\circledcirc$}
\BinaryInfC{$\varphi_{0}\wedge\ldots\wedge\varphi_{m}\wedge\varpi$}
\UnaryInfC{$(\varphi_{0}\wedge\ldots\wedge\varphi_{m}\wedge\varpi)^{\circledcirc}$}
 \LeftLabel{$\bs{\Pi_{3}}$}
\UnaryInfC{$(\varphi_{0}\wedge\ldots\wedge\varphi_{m}\wedge\varpi)^{\circledcirc} \vee (\varphi_{0}\wedge\ldots\wedge\varphi_{m}\wedge\vartheta)^{\circledcirc}$}
\end{prooftree}
$\bs{\Pi_{4}}$ is analogous.
\end{landscape}

\begin{enumerate}
\setcounter{enumi}{1}
\item[$\phantom{.}$]
\begin{itemize}
\item $\Upsilon$ is closed by derivation. Given a derivation $\m{D}$ of wff in $\bs{F}_{w}$ that concludes $\varpi$ and that has a new neighbourhood variable $N$ like described above. Then there is $i \in \mathbb{N}$, such that $\Upsilon_{i} = \{\varphi_{0},\ldots,\varphi_{m}\}$ contains all hypothesis in $\m{D}$. It means that $(\varphi_{0}\wedge\ldots\wedge\varphi_{m})^{\circledcirc} \in \Gamma'$ and, by the following derivation, we know that $(\varphi_{0}\wedge\ldots\wedge\varphi_{m}\wedge\varpi)^{\circledcirc} \in \Gamma'$ because $\Gamma'$ is a prime theory. If $\rho_{k} = \varpi$ and $k \leq i$, then $\varpi \in \Upsilon$, because we can obtain $(\varphi_{0}\wedge\ldots\wedge\varphi_{j}\wedge\varpi)^{\circledcirc} \in \Gamma'$, where $\Upsilon_{k-1} = \{\varphi_{0},\ldots,\varphi_{j}\}$, by eliminating the extra formulas. If $k > i$, then, from $(\varphi_{0}\wedge\ldots\wedge\varphi_{j})^{\circledcirc}$, we can obtain, by a similar derivation, $(\varphi_{0}\wedge\ldots\wedge\varphi_{j}\wedge\varpi)^{\circledcirc}$ because all the wff of $\Upsilon_{i}$ are in $\Upsilon_{k}$.
\end{itemize}
\end{enumerate}\hspace*{3cm}\begin{footnotesize}\AxiomC{$(\varphi_{0}\wedge\ldots\wedge\varphi_{m})^{\circledcirc}$}
\UnaryInfC{$(\varphi_{0}\wedge\ldots\wedge\varphi_{m})^{\circledcirc}$} \RightLabel{$\circledcirc$}
\UnaryInfC{$\varphi_{0}\wedge\ldots\wedge\varphi_{m}$}
\AxiomC{$(\varphi_{0}\wedge\ldots\wedge\varphi_{m})^{\circledcirc}$}
\UnaryInfC{$(\varphi_{0}\wedge\ldots\wedge\varphi_{m})^{\circledcirc}$} \RightLabel{$\circledcirc$}
\UnaryInfC{$\varphi_{0}\wedge\ldots\wedge\varphi_{m}$}
\AxiomC{[$\varphi_{0}\wedge\ldots\wedge\varphi_{m}$]} \RightLabel{$N$}
\UnaryInfC{$\varphi_{0}\wedge\ldots\wedge\varphi_{m}$}
\AxiomC{[$\varphi_{0}\wedge\ldots\wedge\varphi_{m}$]} \RightLabel{$N$}
\UnaryInfC{$\m{D}'$} \RightLabel{$N$}
\UnaryInfC{$\varpi$} \RightLabel{$N$}
\BinaryInfC{$\varphi_{0}\wedge\ldots\wedge\varphi_{m}\wedge\varpi$} \RightLabel{$\circledcirc$}
\BinaryInfC{$\varphi_{0}\wedge\ldots\wedge\varphi_{m}\wedge\varpi$} \RightLabel{$\circledcirc$}
\BinaryInfC{$\varphi_{0}\wedge\ldots\wedge\varphi_{m}\wedge\varpi$}
\UnaryInfC{$(\varphi_{0}\wedge\ldots\wedge\varphi_{m}\wedge\varpi)^{\circledcirc}$}\DisplayProof\end{footnotesize}

\begin{enumerate}
\item[$\phantom{.}$]
\begin{itemize}
\item[$\phantom{.}$] The derivation $\m{D}'$ is obtained from $\m{D}$ by 1) adding $\varphi_{0}\wedge\ldots\wedge\varphi_{m}$ in the place of the hypothesis $\varphi_{i}$ and the following $\wedge$-eliminations to recover $\varphi_{i}$; 2) binding the variable $N$ added to $\m{D}$.
\end{itemize}
\end{enumerate}

\begin{enumerate}
\setcounter{enumi}{1}
\item $\Upsilon$ is accepted by $\Gamma'$: for every wff $\varpi \in \Upsilon$, there is $i \in \mathbb{N}$, such that $\varpi \in \Upsilon_{i}$, what means that there is a conjunction such that $(\varphi_{0}\wedge\ldots\wedge\varphi_{m}\wedge\varpi)^{\circledcirc} \in \Gamma'$. We obtain $\varpi^{\circledcirc} \in \Gamma'$ by a derivation with some $\wedge$-eliminations.

\item The wff of the form $\gamma^{\Theta,\bullet} \in \Upsilon$ demand a prime n-theory to represent a world in which $\gamma^{\Theta}$ holds. The argument is, like in the classical situation, analogous to the case of the wff $\beta^{\Omega,\circledcirc} \in \Gamma'$.

\item The wff of the form $\gamma^{\Theta,\ast} \in \Upsilon$ do not require the existence of a prime n-theory to represent a world in which $\gamma^{\Theta}$ holds, on the contrary, it only requires that there is no prime n-theory, accepted by $\Upsilon$, in which $\gamma^{\Theta}$ does not hold.
\end{enumerate}\end{proof}

\begin{corollary} \label{I_trick}
$\Gamma \not\vdash \alpha^{\Sigma}$ iff there is a model $\m{M}$, such that $\m{M} \models \phi^{\Theta}$, for every $\phi^{\Theta} \in \Gamma$, and $\m{M} \not\models \alpha^{\Sigma}$.
\end{corollary}
\begin{proof}
$\Gamma \not\vdash \alpha^{\Sigma}$ iff $\Gamma \cup \{\neg\alpha^{\Sigma}\}$ is n-consistent by lemma \ref{I_consistency} and the definition of n-consistent set. By lemmas \ref{I_modelForSet} and \ref{I_consistentModel}, $\Gamma \cup \{\neg\alpha^{\Sigma}\}$ is n-consistent iff there is a model $\m{M}$, such that $\m{M} \models \phi^{\Theta}$, for every $\phi^{\Theta} \in \Gamma \cup \{\neg\alpha^{\Sigma}\}$. It means that $\m{M}$ satisfies every formula of $\Gamma$ and $\m{M} \not\models \alpha^{\Sigma}$.\end{proof}

\begin{theorem}
$\Gamma \models \alpha^{\Sigma}$ implies $\Gamma \vdash \alpha^{\Sigma}$ (Completeness).
\end{theorem}

\begin{proof}
$\Gamma \not\vdash \alpha^{\Sigma}$ implies $\Gamma \not\models \alpha^{\Sigma}$, by the corollary \ref{I_trick} and the definition of logical consequence.\end{proof}

\begin{theorem} \label{I_normalization}
Every derivation of iPUC-Logic is normalizable.
\end{theorem}

\begin{proof}
To the PUC normalization arguments we add some permutation reductions following the approach of van Dalen \cite{vanDalen} for the operators not in $\m{L}_{-}$.\end{proof}

The arguments about decidability and complexity of iPUC are analogous to PUC and produce analogous results.

\begin{theorem}
The problem of satisfiability is $\bs{PS}$-complete for iPUC-Logic.
\end{theorem}

\begin{proof}
We repeat the arguments for theorem 57 of \cite{PUC-Logic-ArXiv}, but with reference to the article of Statman\cite{Statman}.
\end{proof}

\section{Constructive Counterfactuals}

In the proof of the relative completeness of PUC-Logic for the Lewis V-Logic, the proofs of CONNEX and CPR demanded the classical absurd rule. In order to do provide a constructive approach over the counterfactual logic, we need to find another way to prove them. Lewis \cite{Lewis} gave us (page 124) an alternative formulation for the CPR rule with a simpler rule and an axiom schema: $$\AxiomC{$\phi \ra \psi$}
\UnaryInfC{$\psi \cp \phi$} \DisplayProof \mbox{ and }(\phi\cp(\phi\vee\psi))\vee(\psi\cp(\phi\vee\psi))$$ That rule is a derived rule in iPUC as we can see by the derivation below, using lemma 59 of \cite{PUC-Logic-ArXiv} (proof $\Pi$):\begin{center}
\AxiomC{$^{1}$[$\phi^{\bullet}$]} \RightLabel{$N$}
\UnaryInfC{$\phi^{\bullet}$} \RightLabel{$N,\bullet$}
\UnaryInfC{$\phi$}
\AxiomC{$^{2}$[$\phi$]} \RightLabel{$N,u$}
\UnaryInfC{$\phi$}
\AxiomC{$\Pi$} \RightLabel{$N,u$}
\UnaryInfC{$\phi \ra \psi$} \RightLabel{$N,u$}
\BinaryInfC{$\psi$} \RightLabel{$N,u$}
\UnaryInfC{$\psi$} \RightLabel{$N,\bullet$}
\UnaryInfC{$\psi$} \RightLabel{$N$}
\UnaryInfC{$\psi^{\bullet}$} \RightLabel{$N$}
\BinaryInfC{$\psi^{\bullet}$} \RightLabel{$N$} \LeftLabel{1}
\UnaryInfC{$\phi^{\bullet} \ra \psi^{\bullet}$} \RightLabel{$\circledast$}
\UnaryInfC{$\phi^{\bullet} \ra \psi^{\bullet}$}
\UnaryInfC{$(\phi^{\bullet} \ra \psi^{\bullet})^{\circledast}$} \DisplayProof
\end{center}

Inside iPUC-Logic, the axioms $(\phi\cp(\phi\vee\psi))\vee(\psi\cp(\phi\vee\psi))$ are unnecessary to recover the completeness of the V-Logic, because they are derived formulas in the presence of CONNEX:

\begin{scriptsize}
\begin{prooftree}
\AxiomC{$(\phi^{\bullet} \ra \psi^{\bullet})^{\circledast} \vee (\psi^{\bullet} \ra \phi^{\bullet})^{\circledast}$}
\UnaryInfC{$(\phi^{\bullet} \ra \psi^{\bullet})^{\circledast} \vee (\psi^{\bullet} \ra \phi^{\bullet})^{\circledast}$}
\AxiomC{$^{3}$[$(\phi^{\bullet} \ra \psi^{\bullet})^{\circledast}$]}
\UnaryInfC{$(\phi^{\bullet} \ra \psi^{\bullet})^{\circledast}$} \RightLabel{$\circledast$}
\UnaryInfC{$\phi^{\bullet} \ra \psi^{\bullet}$}
\AxiomC{$^{1}$[$\phi^{\bullet}$]} \RightLabel{$\circledast$}
\UnaryInfC{$\phi^{\bullet}$}
\BinaryInfC{$\psi^{\bullet}$} \RightLabel{$\circledast,\bullet$}
\UnaryInfC{$\psi$} \RightLabel{$\circledast,\bullet$}
\UnaryInfC{$\phi \vee \psi$} \RightLabel{$\circledast$}
\UnaryInfC{$(\phi \vee \psi)^{\bullet}$} \RightLabel{$\circledast$} \LeftLabel{1}
\UnaryInfC{$\phi^{\bullet} \ra (\phi \vee \psi)^{\bullet}$}
\UnaryInfC{$((\phi^{\bullet} \ra (\phi \vee \psi)^{\bullet})^{\circledast}$}
\UnaryInfC{$((\phi^{\bullet} \ra (\phi \vee \psi)^{\bullet})^{\circledast} \vee ((\psi^{\bullet} \ra (\phi \vee \psi)^{\bullet})^{\circledast}$}
\AxiomC{$^{3}$[$(\psi^{\bullet} \ra \phi^{\bullet})^{\circledast}$]}
\UnaryInfC{$(\psi^{\bullet} \ra \phi^{\bullet})^{\circledast}$} \RightLabel{$\circledast$}
\UnaryInfC{$\psi^{\bullet} \ra \phi^{\bullet}$}
\AxiomC{$^{2}$[$\psi^{\bullet}$]} \RightLabel{$\circledast$}
\UnaryInfC{$\psi^{\bullet}$}
\BinaryInfC{$\phi^{\bullet}$} \RightLabel{$\circledast,\bullet$}
\UnaryInfC{$\phi$} \RightLabel{$\circledast,\bullet$}
\UnaryInfC{$\phi \vee \psi$} \RightLabel{$\circledast$}
\UnaryInfC{$(\phi \vee \psi)^{\bullet}$} \RightLabel{$\circledast$} \LeftLabel{2}
\UnaryInfC{$\psi^{\bullet} \ra (\phi \vee \psi)^{\bullet}$}
\UnaryInfC{$((\psi^{\bullet} \ra (\phi \vee \psi)^{\bullet})^{\circledast}$}
\UnaryInfC{$((\phi^{\bullet} \ra (\phi \vee \psi)^{\bullet})^{\circledast} \vee ((\psi^{\bullet} \ra (\phi \vee \psi)^{\bullet})^{\circledast}$} \LeftLabel{3}
\TrinaryInfC{$((\phi^{\bullet} \ra (\phi \vee \psi)^{\bullet})^{\circledast} \vee ((\psi^{\bullet} \ra (\phi \vee \psi)^{\bullet})^{\circledast}$}
\end{prooftree}
\end{scriptsize}

So, in order to recover the full V-Logic expressivity, we need to introduce some rules into the iPUC-Logic. The first approach is a rule, motivated by \cite{Escobar}, with two restrictions: (a) the formulas fit into their contexts; (b) $\Delta$ has no universal quantifier. And the CONNEX axioms became theorems:
\begin{prooftree}
\AxiomC{[$\alpha^{\Sigma,\bullet} \ra \beta^{\Omega,\bullet}$]} \RightLabel{$\Delta,\circledast$}
\alwaysSingleLine
\UnaryInfC{$\Pi_{1}$} \RightLabel{$\Theta$}
\UnaryInfC{$\gamma^{\Lambda}$}
\AxiomC{[$\beta^{\Omega,\bullet} \ra \alpha^{\Sigma,\bullet}$]} \RightLabel{$\Delta,\circledast$}
\alwaysSingleLine
\UnaryInfC{$\Pi_{2}$} \RightLabel{$\Theta$}
\UnaryInfC{$\gamma^{\Lambda}$} \LeftLabel{31:\hspace*{1cm}} \RightLabel{$\Theta$}
\alwaysSingleLine
\BinaryInfC{$\gamma^{\Lambda}$}
\end{prooftree}

\begin{center}
\AxiomC{[$\alpha^{\Sigma,\bullet} \ra \beta^{\Omega,\bullet}$]} \RightLabel{$\circledast$}
\UnaryInfC{$\alpha^{\Sigma,\bullet} \ra \beta^{\Omega,\bullet}$}
\UnaryInfC{$(\alpha^{\Sigma,\bullet} \ra \beta^{\Omega,\bullet})^{\circledast}$}
\UnaryInfC{$(\alpha^{\Sigma,\bullet} \ra \beta^{\Omega,\bullet})^{\circledast} \vee (\beta^{\Omega,\bullet} \ra \alpha^{\Sigma,\bullet})^{\circledast}$}
\AxiomC{[$\beta^{\Omega,\bullet} \ra \alpha^{\Sigma,\bullet}$]} \RightLabel{$\circledast$}
\UnaryInfC{$\beta^{\Omega,\bullet} \ra \alpha^{\Sigma,\bullet}$}
\UnaryInfC{$(\beta^{\Omega,\bullet} \ra \alpha^{\Sigma,\bullet})^{\circledast}$}
\UnaryInfC{$(\alpha^{\Sigma,\bullet} \ra \beta^{\Omega,\bullet})^{\circledast} \vee (\beta^{\Omega,\bullet} \ra \alpha^{\Sigma,\bullet})^{\circledast}$}
\BinaryInfC{$(\alpha^{\Sigma,\bullet} \ra \beta^{\Omega,\bullet})^{\circledast} \vee (\beta^{\Omega,\bullet} \ra \alpha^{\Sigma,\bullet})^{\circledast}$} \DisplayProof
\end{center}

\begin{definition}
Given any model $\m{M} = \ob \m{W} , \$ , \m{V} , \chi \cb$, the set of testimonials of $\alpha^{\Sigma}$ is the set of neighbourhoods $\m{T}(\alpha^{\Sigma}) = \{ N \in \$(\chi) \; | \; \ob \m{W} , \$ , \m{V} , \chi , N \cb \models \alpha^{\Sigma,\bullet} \}$ and the set of believers of $\alpha^{\Sigma}$ is the set of neighbourhoods $\m{B}(\alpha^{\Sigma}) = \{ N \in \$(\chi) \; | \; \ob \m{W} , \$ , \m{V} , \chi , N \cb \models \alpha^{\Sigma,\ast} \}$.
\end{definition}

$\m{T}(\alpha^{\Sigma})$ is an hereditary set \cite{Goldblatt} in $\$(\chi)$, because given any neighbourhood $M \in \m{T}$, for every $L \in \$(\chi)$, such that $M \subset L$, then $\ob \m{W} , \$ , \m{V} , \chi , L \cb \models \shneg M$ and, by world existential propagation, $\ob \m{W} , \$ , \m{V} , \chi , L \cb \models \alpha^{\Sigma,\bullet}$ and $L \in \m{T}$. But the hereditary sets are in total order in $\$(\chi)$, so, given any set $\m{T}(\beta^{\Omega})$, either $\m{T}(\alpha^{\Sigma}) \subset \m{T}(\beta^{\Omega})$ or $\m{T}(\beta^{\Omega}) \subset \m{T}(\alpha^{\Sigma})$. In the case where $\m{T}(\alpha^{\Sigma}) \subset \m{T}(\beta^{\Omega})$, we know that every neighbourhood, that has a $\alpha^{\Sigma}$-world, must have a $\beta^{\Omega}$-world and we conclude that $\m{M} \models (\alpha^{\Sigma,\bullet} \ra \beta^{\Omega,\bullet})^{\circledast}$. By analogy, in the other case $\m{M} \models (\beta^{\Omega,\bullet} \ra \alpha^{\Sigma,\bullet})^{\circledast}$. So, by definition, $\m{M} \models (\alpha^{\Sigma,\bullet} \ra \beta^{\Omega,\bullet})^{\circledast} \vee (\beta^{\Omega,\bullet} \ra \alpha^{\Sigma,\bullet})^{\circledast}$ and the rule 31 is sound.\\

We can see that the soundness of rule 31 relies upon the hereditary sets definition. But the rule depends on the index of two formulas, the operator $\ra$ and the scope of the rule and it makes this rule two much complex. So, instead of adding rule 31 into iPUC-Logic, we prefer to add the representation of the hereditary sets into the system. We can also treat the hereditary sets of believers to recover some symmetry between the world quantifiers, because we have, by a similar argument, $\m{M} \models (\alpha^{\Sigma,\ast} \ra \beta^{\Omega,\ast})^{\circledast} \vee (\beta^{\Omega,\ast} \ra \alpha^{\Sigma,\ast})^{\circledast}$. This extension of iPUC will be called iPUC$^{V}$-Logic.

\begin{figure}[htbp] \label{I_theSystem}
\begin{dedsystem}
\hline

\AxiomC{$\phantom{-}$} \LeftLabel{32:} \RightLabel{$\Delta,\m{T}(\alpha^{\Sigma})$}
\UnaryInfC{$\alpha^{\Sigma,\bullet}$} 
\alwaysNoLine \UnaryInfC{$\phantom{-}$} \DisplayProof &

\AxiomC{$\phantom{-}$} \RightLabel{$\Delta,\m{T}(\beta^{\Omega})$}
\UnaryInfC{$\alpha^{\Sigma}$} \LeftLabel{33:} \RightLabel{$\Delta,\circledast$}
\UnaryInfC{$\beta^{\Omega,\bullet} \ra \alpha^{\Sigma}$}
\alwaysNoLine \UnaryInfC{$\phantom{-}$} \DisplayProof &

\AxiomC{$\phantom{-}$} \RightLabel{$\Delta,\circledast$}
\UnaryInfC{$\beta^{\Omega,\bullet} \ra \alpha^{\Sigma}$} \LeftLabel{34:}  \RightLabel{$\Delta,\m{T}(\beta^{\Omega})$}
\UnaryInfC{$\alpha^{\Sigma}$}
\alwaysNoLine \UnaryInfC{$\phantom{-}$} \DisplayProof \\

\AxiomC{[$\alpha^{\Sigma,\bullet}$]} \RightLabel{$\Delta,\m{T}(\beta^{\Omega})$}
\alwaysSingleLine
\UnaryInfC{$\Pi_{1}$} \RightLabel{$\Theta$}
\UnaryInfC{$\gamma^{\Psi}$}
\AxiomC{[$\beta^{\Omega,\bullet}$]} \RightLabel{$\Delta,\m{T}(\alpha^{\Sigma})$}
\alwaysSingleLine
\UnaryInfC{$\Pi_{2}$} \RightLabel{$\Theta$}
\UnaryInfC{$\gamma^{\Psi}$} \LeftLabel{35:} \RightLabel{$\Theta$}
\alwaysSingleLine
\BinaryInfC{$\gamma^{\Psi}$}
\alwaysNoLine \UnaryInfC{$\phantom{-}$} \DisplayProof &

\AxiomC{$\phantom{-}$} \LeftLabel{36:} \RightLabel{$\Delta,\m{B}(\alpha^{\Sigma})$}
\UnaryInfC{$\alpha^{\Sigma,\ast}$} 
\alwaysNoLine \UnaryInfC{$\phantom{-}$} \DisplayProof &

\AxiomC{$\phantom{-}$} \RightLabel{$\Delta,\m{B}(\beta^{\Omega})$}
\UnaryInfC{$\alpha^{\Sigma}$} \LeftLabel{37:} \RightLabel{$\Delta,\circledast$}
\UnaryInfC{$\beta^{\Omega,\ast} \ra \alpha^{\Sigma}$}
\alwaysNoLine \UnaryInfC{$\phantom{-}$} \DisplayProof \\

\AxiomC{$\phantom{-}$} \RightLabel{$\Delta,\circledast$}
\UnaryInfC{$\beta^{\Omega,\ast} \ra \alpha^{\Sigma}$} \LeftLabel{38:}  \RightLabel{$\Delta,\m{B}(\beta^{\Omega})$}
\UnaryInfC{$\alpha^{\Sigma}$}
\alwaysNoLine \UnaryInfC{$\phantom{-}$} \DisplayProof &

\AxiomC{[$\alpha^{\Sigma,\ast}$]} \RightLabel{$\Delta,\m{B}(\beta^{\Omega})$}
\alwaysSingleLine
\UnaryInfC{$\Pi_{1}$} \RightLabel{$\Theta$}
\UnaryInfC{$\gamma^{\Psi}$}
\AxiomC{[$\beta^{\Omega,\ast}$]} \RightLabel{$\Delta,\m{B}(\alpha^{\Sigma})$}
\alwaysSingleLine
\UnaryInfC{$\Pi_{2}$} \RightLabel{$\Theta$}
\UnaryInfC{$\gamma^{\Psi}$} \LeftLabel{39:} \RightLabel{$\Theta$}
\alwaysSingleLine
\BinaryInfC{$\gamma^{\Psi}$}
\alwaysNoLine \UnaryInfC{$\phantom{-}$} \DisplayProof &\\
 
\hline
\end{dedsystem}
\caption{The additional rules of iPUC$^{V}$-ND}
\end{figure}

\noindent Restriction for rules from 32 to 39: All formulas must fit into their contexts. Restriction for the rules 32 and 39: $\Delta$ has no universal quantifier.

\begin{lemma}
iPUC$^{V}$ is sound.
\end{lemma}
\begin{proof} First of all, we need to add $\m{T}$ and $\m{B}$ as labels. Those sets may be empty and may be defined by an universal quantification. For example, $\m{T}(\alpha^{\Sigma})$ may be defined by the sentence: for all $N \in \$(\chi)$, such that $\alpha^{\Sigma}$ holds in some of its worlds. So, we consider $\m{T}$ and $\m{B}$ as restrictions of the label $\circledast$, thus as universal quantifiers over neighbourhoods. For the semantics:
\begin{enumerate}
\setcounter{enumi}{18}
\item $\ob \m{W} , \$ , \m{V} , \chi \cb \models \beta^{\Omega,\m{T}(\alpha^{\Sigma})}$ iff:
$\alpha^{\Sigma} \in \bs{F}_{n}$, $\beta^{\Omega} \in \bs{F}_{w}$ and $\forall \lambda \in \m{A}(\chi)$, $\forall N \in \$(\lambda)$ : if $\ob \m{W} , \$ , \m{V} , \lambda , N \cb \models \alpha^{\Sigma,\bullet}$, then $\ob \m{W} , \$ , \m{V} , \lambda , N \cb \models \beta^{\Omega}$;
\item $\ob \m{W} , \$ , \m{V} , \chi \cb \models \beta^{\Omega,\m{B}(\alpha^{\Sigma})}$ iff: $\alpha^{\Sigma} \in \bs{F}_{n}$, $\beta^{\Omega} \in \bs{F}_{w}$ and $\forall \lambda \in \m{A}(\chi)$, $\forall N \in \$(\lambda)$ : if $\ob \m{W} , \$ , \m{V} , \lambda , N \cb \models \alpha^{\Sigma,\ast}$, then $\ob \m{W} , \$ , \m{V} , \lambda , N \cb \models \beta^{\Omega}$.
\end{enumerate}
For the definition of the previous rules, we only need to change the first restriction of the rule 9: (a) $\Delta$ must have no occurrence of universal quantifiers over neighbourhoods. For the proof of the lemmas and theorems for soundness and completeness, the arguments for $\m{T}$ and $\m{B}$ should follow the arguments for $\circledast$.\\To prove that iPUC$^{V}$ is sound, we follow the strategy for PUC and iPUC. So, we need to prove that the additional rules preserve resolution.
\begin{enumerate}
\setcounter{enumi}{31}
\item From the fitting relation and lemma 2 of \cite{PUC-Logic-ArXiv}, $s(\{\Delta,\m{T}(\alpha^{\Sigma}\})$ must be odd, because $\alpha^{\Sigma,\bullet} \in \bs{F}_{w}$. So, $s(\Delta)$ is even. If we take some model $\m{H} = \ob \m{W} , \$ , \m{V} , z \cb$, such that $\m{M} \multimap_{s(\Delta)} \m{H}$, then we know, by definition, that $\forall \lambda \in \m{A}(\chi)$, $\forall N \in \$(\lambda)$ : if $\ob \m{W} , \$ , \m{V} , \lambda , N \cb \models \alpha^{\Sigma,\bullet}$, then $\ob \m{W} , \$ , \m{V} , \lambda , N \cb \models \alpha^{\Sigma}$, which means that $\m{M} \models^{\Delta} \alpha^{\Sigma,\bullet,\m{T}(\alpha^{\Sigma})}$ and, by rule 13, $\m{M} \models^{\Delta,\m{T}(\alpha^{\Sigma})} \alpha^{\Sigma,\bullet}$;

\item If $\m{M} \models^{\Delta,\m{T}(\beta^{\Omega})} \alpha^{\Sigma}$, then, by the rule 14, $\m{M} \models^{\Delta} \alpha^{\Sigma,\m{T}(\beta^{\Omega})}$. From $\alpha^{\Sigma,\m{T}(\beta^{\Omega})} \in \bs{F}_{n}$, the fitting relation and lemma 2 of \cite{PUC-Logic-ArXiv}, we know that $s(\Delta)$ is even. If we take some model $\m{H} = \ob \m{W} , \$ , \m{V} , z \cb$, such that $\m{M} \multimap_{s(\Delta)} \m{H}$ and $\m{H} \models \alpha^{\Sigma,\m{T}(\beta^{\Omega})}$, then $\forall \lambda \in \m{A}(z)$, $\forall N \in \$(\lambda)$ : if $\ob \m{W} , \$ , \m{V} , \lambda , N \cb \models \beta^{\Omega,\bullet}$, then $\ob \m{W} , \$ , \m{V} , \lambda , N \cb \models \alpha^{\Sigma}$. By definition $\forall \lambda \in \m{A}(z) : \$(\lambda) = \$(z)$, then we know that $\forall N \in \$(z) : \forall \lambda \in \m{A}(z)$, if $\ob \m{W} , \$ , \m{V} , \lambda , N \cb \models \beta^{\Omega,\bullet}$, then $\ob \m{W} , \$ , \m{V} , \lambda , N \cb \models \alpha^{\Sigma}$, what means that $\forall N \in \$(z) : \beta^{\Omega,\bullet} \ra \alpha^{\Sigma}$. We conclude that $\m{H} \models (\beta^{\Omega,\bullet} \ra \alpha^{\Sigma})^{\circledast}$ and, by definition, $\alpha^{\Sigma,\m{T}(\beta^{\Omega})} \models_{\m{M}:s(\Delta)} (\beta^{\Omega,\bullet} \ra \alpha^{\Sigma})^{\circledast}$, which means, by lemma 21 of \cite{PUC-Logic-ArXiv}, that $\m{M} \models^{\Delta} (\beta^{\Omega,\bullet} \ra \alpha^{\Sigma})^{\circledast}$ and, by rule 13, $\m{M} \models^{\Delta,\circledast} \beta^{\Omega,\bullet} \ra \alpha^{\Sigma}$;

\item Analogous to the rule 33, but in the reverse order of the conclusions;

\item From rule 14, the fitting relation, and lemma 2 of \cite{PUC-Logic-ArXiv}, we know that $s(\Delta)$ is even. If we take some model $\m{H} = \ob \m{W} , \$ , \m{V} , z \cb$, such that $\m{M} \multimap_{s(\Delta)} \m{H}$, we know that $\m{T}(\alpha^{\Sigma}) \subset \m{T}(\beta^{\Omega})$ or $\m{T}(\beta^{\Omega}) \subset \m{T}(\alpha^{\Sigma})$. This means that $\forall N \in \$(z)$ : if $\ob \m{W} , \$ , \m{V} , z , N \cb \models \alpha^{\Sigma,\bullet}$ then $\ob \m{W} , \$ , \m{V} , z , N \cb \models \beta^{\Omega,\bullet}$ or $\forall N \in \$(z)$ : if $\ob \m{W} , \$ , \m{V} , z , N \cb \models \beta^{\Omega,\bullet}$ then $\ob \m{W} , \$ , \m{V} , z , N \cb \models \alpha^{\Sigma,\bullet}$. We can expressed it by $\m{H} \models \alpha^{\Sigma,\bullet,\m{T}(\beta^{\Omega})} \vee \beta^{\Omega,\bullet,\m{T}(\alpha^{\Sigma})}$, because $\forall \lambda \in \m{A}(z) : \$(\lambda) = \$(z)$. By definition, $\m{H} \models \alpha^{\Sigma,\bullet,\m{T}(\beta^{\Omega})}$ or $\m{H} \models \beta^{\Omega,\bullet,\m{T}(\alpha^{\Sigma})}$, then, by definition, $\m{M} \models^{\Delta} \alpha^{\Sigma,\bullet,\m{T}(\beta^{\Omega})}$ or $\m{M} \models^{\Delta} \beta^{\Omega,\bullet,\m{T}(\alpha^{\Sigma})}$ and, using rule 13, $\m{M} \models^{\Delta,\m{T}(\beta^{\Omega})} \alpha^{\Sigma,\bullet}$ or $\m{M} \models^{\Delta,\m{T}(\alpha^{\Sigma})} \beta^{\Omega,\bullet}$. Using the same argument for the subderivations $\Pi_{1}$ and $\Pi_{2}$ as in lemma 28 of \cite{PUC-Logic-ArXiv}, then $\m{M} \models^{\Theta} \gamma^{\Psi}$ and the hypothesis may be discharged;

\item Analogous to the case of rule 32;

\item Analogous to the case of rule 33;

\item Analogous to the rule 37, but in the reverse order of the conclusions;

\item Analogous to the case of rule 35.
\end{enumerate}\end{proof}

\begin{lemma}
iPUC$^{V}$ is normalizing.
\end{lemma}

\begin{proof}
We consider only the additional rules in this proof. The rules 32 and 36 introduce no maximum formula and no detour, because they take no hypothesis. The rules 35 and 39 introduce no maximum formula and no detour, because they do not change the formulas and because they left no open hypothesis to produce a detour. The rules 33 and 34 can only produce a maximum formula if the conclusion of the rule 33 is taken as hypothesis of the rule 34. This situation can also be seen as a detour and is easily removed by the following reduction rule:
\begin{center}\begin{tabular}{lcr}
\AxiomC{$\phantom{-}$} \RightLabel{$\Delta,\m{T}(\beta^{\Omega})$}
\UnaryInfC{$\alpha^{\Sigma,\bullet}$} \LeftLabel{33:} \RightLabel{$\Delta,\circledast$}
\UnaryInfC{$\beta^{\Omega,\bullet} \ra \alpha^{\Sigma,\bullet}$} \LeftLabel{34:}  \RightLabel{$\Delta,\m{T}(\beta^{\Omega})$}
\UnaryInfC{$\alpha^{\Sigma,\bullet}$}
\UnaryInfC{$\Pi$}
\alwaysNoLine \UnaryInfC{$\phantom{-}$} \DisplayProof&$\rhd$&\AxiomC{$\phantom{-}$} \RightLabel{$\Delta,\m{T}(\beta^{\Omega})$}
\UnaryInfC{$\alpha^{\Sigma,\bullet}$}
\UnaryInfC{$\Pi$}
\alwaysNoLine \UnaryInfC{$\phantom{-}$} \DisplayProof
\end{tabular}\end{center}
\noindent The case of rules 37 and 38 is similar to the case of rules 33 and 34.\end{proof}

\newpage

\begin{lemma}
iPUC$^{V}$ is complete for $\bs{V}$-Logic.
\end{lemma}

\begin{proof}
\begin{prooftree}\AxiomC{$^{1}$[$\alpha^{\bullet}$]} \RightLabel{$\m{T}(\beta)$}
\UnaryInfC{$\alpha^{\bullet}$} \RightLabel{$\circledast$}
\UnaryInfC{$\beta^{\bullet} \ra \alpha^{\bullet}$}
\UnaryInfC{$(\beta^{\bullet} \ra \alpha^{\bullet} )^{\circledast}$}
\UnaryInfC{$( \alpha^{\bullet} \ra \beta^{\bullet} )^{\circledast} \vee ( \beta^{\bullet} \ra \alpha^{\bullet} )^{\circledast}$}
\AxiomC{$^{1}$[$\beta^{\bullet}$]} \RightLabel{$\m{T}(\alpha)$}
\UnaryInfC{$\beta^{\bullet}$} \RightLabel{$\circledast$}
\UnaryInfC{$\alpha^{\bullet} \ra \beta^{\bullet}$}
\UnaryInfC{$(\alpha^{\bullet} \ra \beta^{\bullet} )^{\circledast}$}
\UnaryInfC{$( \alpha^{\bullet} \ra \beta^{\bullet} )^{\circledast} \vee ( \beta^{\bullet} \ra \alpha^{\bullet} )^{\circledast}$} \LeftLabel{\scriptsize{1}}
\BinaryInfC{$( \alpha^{\bullet} \ra \beta^{\bullet} )^{\circledast} \vee ( \beta^{\bullet} \ra \alpha^{\bullet} )^{\circledast}$} 
\alwaysNoLine \UnaryInfC{$\phantom{-}$}\end{prooftree}
\noindent By the beginning of this section, we only needed to prove that the CONNEX axioms are theorems.\end{proof}

\section*{Conclusions}

The intuitionistic approach have shown the necessity to introduce a representation of the relation of total order among neighbourhoods. And the total order is a core concept in Lewis work for counterfactual logic and deontic logic because of the role of the nesting function in his arguments. The possibility of expressing the total order in iPUC$^{V}$ may contribute to a deeper comprehension of Lewis arguments.\\

The possibility of a constructive approach over Lewis counterfactuals may also open some philosophical questions for a further research.

\end{document}